\documentclass[11pt]{article}

\usepackage{psfrag,amsmath,amsbsy,amsfonts,amssymb,amsthm,fullpage}
\usepackage{graphicx}
\usepackage{algorithmic,mathtools}
\usepackage{natbib}
\usepackage{color}
\usepackage{tikz}
\usetikzlibrary{calc,shapes}
\usepackage{subfig}
\usepackage{hyperref}
\usepackage{multirow}
\usepackage[linesnumbered,boxed,ruled, vlined]{algorithm2e}
\theoremstyle{plain}

\theoremstyle{definition}

\newtheorem{observation}{Observation}[section]
\theoremstyle{remark}

\newtheorem{theorem}{Theorem}
\newtheorem{lemma}[theorem]{Lemma}
\newtheorem{corollary}[theorem]{Corollary}

\newtheorem{definition}[theorem]{Definition}

\newcommand{\kbnonuniout}{($\{L_u\}$,$\{U_u\}$,$k$,$p$)-\textsc{Center}}

\newcommand{\knonuni}{($L$,$\{U_u\}$,$k$)-\textsc{Center}}

\newcommand{\knonuniout}{($L$,$\{U_u\}$,$k$,$p$)-\textsc{Center}}

\newcommand{\unisoftout}{($L$,$U$,soft-$\emptyset$,$p$)-\textsc{Center}}
\newcommand{\uniout}{($L$,$U$,$\emptyset$,$p$)-\textsc{Center}}

\newcommand{\supnonunisoftout}{($L$,$\{U_u\}$,soft-$\emptyset$,$p$)-\textsc{Supplier}}
\newcommand{\nonuniout}{($L$,$\{U_u\}$,$\emptyset$,$p$)-\textsc{Center}}

\newcommand{\ksupnonuni}{($L$,$\{U_u\}$,$k$)-\textsc{Supplier}}

\newcommand{\ksupbnonuniout}{($\{L_u\}$,$\{U_u\}$,$k$,$p$)-\textsc{Supplier}}

\newcommand{\supnonuniout}{($L$,$\{U_u\}$,$\emptyset$,$p$)-\textsc{Supplier}}
\newcommand{\ksupnonuniout}{($L$,$\{U_u\}$,$k$,$p$)-\textsc{Supplier}}
\newcommand{\supunisoftout}{($L$,$U$,soft-$\emptyset$,$p$)-\textsc{Supplier}}

\newcommand{\CCT}{\textsc{CCT}}
\newcommand{\dist}{\mathrm{d}}
\newcommand{\calC}{{\mathcal C}}
\newcommand{\calF}{{\mathcal F}}

\newcommand{\Z}{{\mathbb{Z}}}
\newcommand{\R}{{\mathbb{R}}}
\newcommand{\bbbr}{{\mathbb{R}}}
\newcommand{\bbbn}{{\mathbb{N}}}
\newcommand{\bbbz}{{\mathbb{Z}}}
\newcommand{\true}{\mathsf{true}}

\newcommand{\topic}[1]{\vspace{0.2cm}\noindent {\bf #1}}

\title{Capacitated Center Problems with Two-Sided Bounds and Outliers}
\author{Hu Ding\footnotemark[1] \and Lunjia Hu\footnotemark[2] \and Lingxiao Huang\footnotemark[2]
\and Jian Li\footnotemark[2]}

\footnotetext[1]{Computer Science and Engineering, Michigan State University, East Lansing, MI, USA}

\footnotetext[2]{Institute for Interdisciplinary Information Sciences, Tsinghua University, Beijing, China}

\begin{document}

\maketitle

\begin{abstract}
In recent years, the capacitated center problems have attracted a lot of research interest. Given a set of vertices $V$, we want to find a subset of vertices $S$, called centers, such that 
the maximum cluster radius is minimized. Moreover, each center in $S$ should satisfy some capacity constraint, which could be an upper or lower bound on the number of vertices it can serve. Capacitated $k$-center problems with one-sided bounds (upper or lower) have been well studied in previous work, and a constant factor approximation was obtained.

We are the first to study the capacitated center problem with both capacity lower and upper bounds (with or without outliers). We assume each vertex has a uniform lower bound and a non-uniform upper bound. For the case of opening exactly $k$ centers, we note that a generalization of a recent LP approach can achieve constant factor approximation algorithms for our problems. Our main contribution is a simple combinatorial algorithm
for the case where there is no cardinality constraint on the 
number of open centers.
Our combinatorial algorithm is simpler and achieves better constant approximation factor compared to the LP approach.
\end{abstract}


\newpage
\section{introduction}
The $k$-center clustering is a fundamental problem in theoretical computer science and has numerous applications in a variety of fields.
Roughly speaking, given a metric space containing a set of vertices, the $k$-center problem asks for a subset of $k$ vertices, called centers,
such that the maximum radius of the induced $k$ clusters is minimized.
Actually $k$-center clustering falls in the umbrella of the general
{\em facility location} problems which have been extensively studied
in the past decades. Many operation and management problems can be modeled as facility location problems, and usually the input vertices and selected centers are also called ``clients" and ``facilities" respectively.
In this paper, we consider a significant generalization of
the $k$-center problem,
where each vertex is associated with a capacity interval;
that is, the cardinality of the resulting cluster centered at the vertex should satisfy the given lower and upper capacity bounds (the formal definition is shown in Section~\ref{sec-pre}).
In addition, we also consider the case where a given number of vertices
may be excluded as outliers.

Besides being a natural combinatorial problem on its own,
the $k$-center problem with both capacity upper and lower bounds
is also strongly motivated by several realistic issues
raised in a variety of application contexts.

\begin{enumerate}
	\item
	In the context of facility location,
	each open facility may be constrained by
	the maximum number of clients it can serve.
	The capacity lower bounds also come naturally,
	since an open facility needs to serve at least a certain number of
	clients in order to generate profit.
	\item
	Several variants of the $k$-center clustering have been used
	in the context of preserving privacy in publication of sensitive data
	(see e.g., \citep{DBLP:journals/talg/AggarwalPFTKKZ10,li2010clustering,sweeney2002k}).
	In such applications, it is important to have
	an appropriate lower bound for the cluster sizes,
	in order to protect the privacy to certain extent
	(roughly speaking, it would be relatively easier for an adversary to identify the clients inside a too small cluster). 	
	\item
	Consider the scenario where the data is distributed over the nodes in a large network. We would like to choose $k$ nodes as central servers,
	and aggregate the information of the entire network.
	We need to minimize the delay (i.e., minimize the cluster radius),
	and at the same time consider the balancedness, for the obvious reason that the machines receiving too much data could be the bottleneck of the system and the ones receiving too little data is not sufficiently energy-efficient~\citep{dick2015data}. 	
\end{enumerate}

Our problem generalizes the classic $k$-center problem as
well as many important variants studied by previous authors.
The optimal approximation results for the classic
$k$-center problem appeared in the 80's:  \cite{gonzalez1985clustering} and \cite{hochbaum1985best} provided a $2$-approximation
in a metric graph; moreover, they proved that any approximation ratio $c<2$ would imply $P=NP$. The first study on capacitated (with only upper bounds) $k$-center clustering is due to \cite{barilan1993allocate} who provided a $10$-approximation algorithm for uniform capacities (i.e., all the upper bounds are identical). Further, \cite{khuller2000capacitated} improved the approximation ratio to be $6$ and $5$ for hard and soft uniform capacities, respectively.
\footnote{
	We can open more than one copies of a facility in the same node
	in the soft capacity version. But in the hard capacity version,
	we can only open at most one copy.
	}
The recent breakthrough for non-uniform (upper) capacities is due to \cite{cygan2012lp}. They developed the first constant approximation algorithm based on LP rounding, though their approximation ratio is about hundreds.
Following this work, \cite{an2015centrality} provided an approximation algorithm with the much lower approximation ratio $9$.
On the imapproximability side, it is impossible to achieve an approximation ratio lower than $3$ for non-uniform capacities unless $P=NP$~\citep{cygan2012lp}.

For the ordinary $k$-center with outliers, a $3$-approximation algorithm was obtained by \cite{charikar2001algorithms}.
\cite{kociumaka2014constant} studied $k$-center with non-uniform upper capacities and outliers, and provided a $25$-approximation algorithm.

$k$-center clustering with lower bounds on cluster sizes was first studied in the context of privacy-preserving data management~\citep{sweeney2002k}. \cite{DBLP:journals/talg/AggarwalPFTKKZ10} provided a $2$-approximation and a $4$-approximation for the cases without and with outliers, respectively.
Further, \cite{ene2013fast} presented a near linear time $(4+\epsilon)$-approximation algorithm in constant dimensional Euclidean space. Note that both \citep{DBLP:journals/talg/AggarwalPFTKKZ10,ene2013fast} are only for uniform lower bounds. Recently,
\cite{ahmadian2016approximation} provided a $3$-approximation and a $5$-approximation for the non-uniform lower bound case without and with outliers.

\noindent\textbf{Our main results.} To the best of our knowledge, we are the first to study the $k$-center with both capacity lower and upper bounds (with or without outliers).
Given a set $V$ of $n$ vertices, we focus on the case where the capacity of each vertex $u\in V$ has a uniform lower bound $L_u=L$ and a non-uniform upper bound $U_u$. Sometimes, we consider a generalized supplier version where we are only allowed to open centers among a facility set $\calF$, see Definition \ref{def:ksupp} for details.
We mainly provide first constant factor approximation algorithms for the following variants, see Table \ref{tab:result} for other results.

\begin{enumerate}
\item \unisoftout\ (Section \ref{sec:simple}): In this problem, both the lower bounds and the upper bounds are uniform, i.e., $L_u=L,U_u=U$ for all $u\in V$. The number of open centers can be arbitrary, i.e., there is no requirement to choose exactly $k$ open centers.
Moreover, we allow multiple open centers at a single vertex $u\in V$ (i.e., soft capacity).
We may exclude $n-p$ outliers.
We provide the first polynomial time combinatorial algorithm
which can achieve an approximate factor of $5$.
\item \nonuniout (Section \ref{sec:generalized}):
In this problem, the lower bounds are uniform, i.e., $L_u=L$ for all $u\in V$, but the upper bound can be nonuniform.
The number of open centers can be arbitrary. We may exclude $n-p$ outliers.
We provide the first polynomial time combinatorial
$11$-approximation for this problem.
\item \knonuni\ (Section \ref{sec:kcenter}): In this problem,
we would like to open exactly $k$ centers, such that the maximum cluster
radius is minimized.
All vertices have the same capacity lower bounds, i.e., $L_u=L$ for all $u\in V$. But the capacity upper bounds may be nonuniform, i.e.,
each vertex $u$ has an individual capacity upper bound $U_u$.
Moreover, we do not exclude any outlier.
We provide the first polynomial time
$9$-approximation algorithm for this problem, based on LP rounding.
\item \knonuniout\ (Section \ref{sec:kcenter}): This problem is the outlier version of the \knonuni\ problem. The problem setting is exactly the same except that we
can exclude $n-p$ vertices as outliers.
We provide a polynomial time $25$-approximation algorithm for this problem.

\end{enumerate}

\begin{table*}[htbp]
  \centering
  \begin{tabular}{|*{4}{c|}}
\hline
 \multicolumn{2}{|c|}{\multirow{2}{*}{Problem Setting}} &
 \multicolumn{2}{c|}{Approximation Ratio}  \\
  \cline{3-4}
  \multicolumn{2}{|c|}{} & Center Version & Supplier Version  \\
\hline\cline{1-4}
  \multirow{4}{*}{Without $k$ Constraint} &
  ($L$,$U$,soft-$\emptyset$,$p$) & 5 & 5 \\\cline{2-4}
  &($L$,$U$,$\emptyset$,$p$) & 10 & 23\\\cline{2-4}
  &($L$,$\{U_u\}$,soft-$\emptyset$,$p$) & 11 & 11 \\\cline{2-4}
  &($L$,$\{U_u\}$,$\emptyset$,$p$) & 11 & 25 \\
  \hline \cline{1-4}
 \multirow{6}{*}{With $k$ Constraint} &
  ($L$,$U$,$k$) & 6 & 9 \\  \cline{2-4}
  & ($L$,$\{U_u\}$,$k$) & 9 & 13 \\  \cline{2-4}
  &($L$,$U$,soft-$k$,$p$)& 13 & 13\\ \cline{2-4}
  &($L$,$U$,$k$,$p$)& 23 & 23 \\\cline{2-4}
  &($L$,$\{U_u\}$,soft-$k$,$p$) & 25 & 25 \\\cline{2-4}
  &($L$,$\{U_u\}$,$k$,$p$)  & 25 & 25 \\\hline
  \end{tabular}
  \vspace{0.1cm}
  \caption{A summarization table for our results in this paper.}
  \label{tab:result}
\end{table*}

\vspace{-0.5cm}

\topic{Our main techniques.}
In Section \ref{sec:greedy}, we consider the first two variants which allow to open arbitrarily many centers.
We design simple and faster
combinatorial algorithms
which can achieve better constant approximation ratios compared to the LP approach.
For the simpler case \unisoftout, we construct a data structure for all possible open centers. We call it a \emph{core-center tree} (\CCT).
Our greedy algorithm mainly contains two procedures. The first procedure \emph{pass-up}
greedily assigns vertices to open centers from the leaves of \CCT\ to the root. After this procedure, there may exist some unassigned vertices around the root. We then introduce the second procedure called \emph{pass-down}, which assigns these vertices in order by finding an \emph{exchange route} each time. For the more general case \nonuniout, our greedy algorithm is similar but somewhat more subtle.
We still construct a \CCT\ and run the pass-up procedure.
Then we obtain an open center set $F$, which may contain redundant centers. However, since we deal with hard capacities and outliers, we need to find a non-redundant open center set which is not 'too far' from $F$ (see Section \ref{sec:generalized} for details) and have enough total capacities. Then by a pass-down procedure, we can assign enough vertices to their nearby open centers.

In Section \ref{sec:algtree} and \ref{sec:kcenter}, we consider the last two variants which require to open exactly $k$ centers.
We generalized the LP approach developed for $k$-center with only
capacity upper bounds~\citep{an2015centrality,kociumaka2014constant} and obtain constant approximation schemes for two-sided capacitated bounds. Due to the lack of space, we defer many details and proofs to a full version.

\subsection{Other Related Work}

The classic $k$-center problem is quite fundamental and has been generalized
in many ways, to incorporate various constraints motivated by
different application scenarios.
Recently, \cite{fernandes2016improved} also provided constant approximations for the fault-tolerant capacitated $k$-center clustering.
\cite{chen2016matroid} studied the matroid center problem
where the selected centers must form an independent set of a given matroid,
and provided constant factor approximation algorithms (with or without outliers).

There is a large body of work on approximation algorithms for the facility location and $k$-median problems (see e.g., \citep{arya2004local,charikar2005improved,charikar1999constant,guha1999greedy,jain2002new,jain2001approximation,korupolu2000analysis,li20131,li2016approximating}). Moreover, \cite{dick2015data} studied multiple balanced clustering problems with uniform capacity intervals, that is, all the lower (upper) bounds are identical; they also consider the problems under the stability assumption.

\subsection{Preliminaries}
\label{sec-pre}
In this paper, we usually work with the following more general problem, called the capacitated $k$-supplier problem.
It is easy to see it generalizes the  capacitated $k$-center problem.
The formal definition is as follows.

\begin{definition}
\label{def:ksupp}
(Capacitated $k$-supplier with two-sided bounds and outliers) Suppose that we have
\begin{enumerate}
\item Two integers $k,p\in \Z_{\geq 0}$;
\item A finite set $\calC$ of clients, and a finite set $\calF$ of facilities;
\item A symmetric distance function $\dist: (\calC\cup \calF)\times (\calC\cup \calF)\rightarrow \R_{\geq 0}$ satisfying the triangle inequality;
\item A capacity interval $[L_u,U_u]$ for each facility $u\in \calF$, where $L_u,U_u\in \Z_{\geq 0}$ and $L_u\leq U_u$.
\end{enumerate}
Our goal is to find a client set $C\subseteq \calC$ of size at least $p$, an open facility set $F\subseteq \calF$ of size exactly $k$, and a function $\phi: C\rightarrow F$ satisfying that $L_u\leq |\phi^{-1}(u)|\leq U_u$ for each $u\in F$, which minimize the maximum cluster radius
$\max_{v\in C} \dist(v, \phi(v))$. If the maximum cluster radius is at most $r$, we call the tuple $(C,F,\phi)$ a distance-$r$ solution.
\end{definition}

We denote the above problem as \ksupbnonuniout.
If the lower bounds are uniform ($L_u=L$ for all $u\in \calF$),
we use $L$ in place $\{L_u\}$, e.g., \ksupnonuni.
Similarly, if the upper bounds are uniform ($U_u=U$ for all $u\in \calF$),
we use $U$ in place $\{U_u\}$.
If there is no constraint to open $k$ centers,
we use $\emptyset$ to replace $k$,
 e.g., \supnonuniout. Also note that the capacitated $k$-center problem with two-sided bounds and outliers is a special case by letting $V=\calC=\calF$, we denote it the \kbnonuniout\ problem.

 By the similar approach of \cite{kociumaka2014constant}, we can reduce the \ksupbnonuniout\ problem to a simpler case. We first introduce some definitions.
 \begin{definition}
 \label{def:induce}
 (Induced distance function) We say the distance function $\dist_G:(\calC\cup\calF)\times(\calC\cup\calF)\rightarrow \bbbr_{\geq 0}$ is induced by an undirected unweighted connected graph $G=(\calC\cup\calF,E)$ if
 \begin{enumerate}
 \item $\forall (u,v)\in E$, we have $u\in\calF$ and  $v\in\calC$.
 \item $\forall a_1,a_2\in\calC\cup\calF$, the distance $\dist_G(a_1,a_2)$ between $a_1$ and $a_2$ equals to the length of the shortest path from $a_1$ to $a_2$.
 \end{enumerate}
 \end{definition}

\begin{definition}
\label{def:bipartite}
(Induced \ksupbnonuniout\ instance)
An \ksupbnonuniout\ instance is called an induced \ksupbnonuniout\ instance if the following properties are satisfied:
\begin{enumerate}
\item The distance function $\dist_G$ is induced by an undirected connected graph $G=(\calC\cup\calF,E)$.
\item The optimal capacitated $k$-supplier value is at most 1.
\end{enumerate}
Moreover, we say this instance is induced by $G$.

\end{definition}

When the graph of interest $G$ is clear from the context, we will use $\dist$ instead of $d_G$ for convenience.
We then show a reduction from solving the generalized  \ksupbnonuniout\ problem to solving induced \ksupbnonuniout\ instances by Lemma \ref{lm:bipartite}. The proof can be found in Appendix \ref{app:lemma4}.

\begin{lemma}
\label{lm:bipartite}
Suppose we have a polynomial time algorithm $A$ that takes as input any induced \ksupbnonuniout\ instance, and outputs a distance-$\rho$ solution. Then, there exists a $\rho$-approximation algorithm for the \ksupbnonuniout\ problem with polynomial running time.
\end{lemma}



By Lemma \ref{lm:bipartite}, we focus on designing an algorithm $A$ for different variants of the induced \ksupbnonuniout\ instances.

\section{Capacitated Center with Two-Sided Bounds and Outliers}
\label{sec:greedy}

In this section, we consider the version that the number of open centers can be arbitrary. By the LP approach in Section \ref{sec:kcenter} and enumerating the number of open centers, we can achieve approximation algorithms for different variants in this case. However, the approximation factor is not small enough. In this section, we introduce a new greedy approach in order to achieve better approximation factors. Since our algorithm is combinatorial, it is easier to be implemented and saves the running time compared to the LP approach.

\subsection{Core-center tree (\CCT)}
\label{sec:cct}
Consider the \supnonuniout\ problem. By Lemma \ref{lm:bipartite}, we only need to consider induced \supnonuniout\ instances induced by an undirected unweighted connected graph $G=(\calC\cup \calF,E)$. We first propose a new data structure called \emph{core-center tree (\CCT)} as follows.

\begin{definition}
\label{def:CCT}
(Core-center tree (\CCT))
Given an induced \supnonuniout\ instance induced by an undirected unweighted connected graph $G=(\calC\cup \calF,E)$, we call a tree $T=(\calF,E_T)$ a core-center tree(\CCT) if the following properties hold.
\begin{enumerate}
\item For each edge $(u,u')\in E_T$, we have $\dist_G(u,u')\leq 2$;
\item Suppose the root of $T$ is at layer 0. Denote $I$ to be the set of vertices in the even layers of $T$. We call $I$ the core-center set of $T$. For any two distinct vertices $u,u'\in I$, we have $\dist_G(u,u')\geq 3$.
\end{enumerate}
\end{definition}

\begin{lemma}
\label{lm:cct}
Given an induced \supnonuniout\ instance induced by an undirected unweighted connected graph $G=(\calC\cup \calF,E)$, we can construct a \CCT\ in polynomial time.
\end{lemma}
\begin{proof}
We first construct a graph $G^2$ on $\calF$ as follows: for each pair $u_1,u_2$ in $\calF$ with distance at most 2, we add an edge $(u_1,u_2)$ in $G^2$. Observe that $G^2$ is connected by Definition \ref{def:induce}. We then construct a spanning tree $T$ of $G^2$ satisfying that all facilities in even layers form an independent set of $G^2$. It is not hard to verify that such a tree is a \CCT. We build $T$ as follows. The above property directly holds from our construction.
\begin{enumerate}
\item Initially, we randomly pick a facility $u\in \calF$ as the root of $T$. We then pick all adjacent facilities of $u$ in $G^2$ as its children (layer 1).
\item By a modified BFS, we continue to construct layer 2 and layer 3. Each time we pick a facility $w$ in layer 1. We iteratively pick an adjacent facility $w'$ of $w$ in $G^2$ which has not been scanned as a child of $w$. After we append $w'$ to layer 2, we immediately pick all unscanned neighbors of $w'$ in $G^2$ as the children of $w'$ (append them in layer 3).
\item We then iteratively construct $T$ until all facilities in $\calF$ have been scanned. Each iteration, we build two consecutive layers: an odd layer and an even layer.
\end{enumerate}
\end{proof}

For any $u\in \calF$, denote $N_G[u]=\{v\in \calC: (u,v)\in E\}$ to be the collection of all neighbors of $u\in \calF$.
\footnote{If $u\in \calC$ is also a client, then $u\in N_G[u]$.}
W.l.o.g., we assume that $U_u\leq |N_G(u)|$ for every facility $u\in \calF$ in this section.
In fact, we can directly delete all $u\in \calF$ satisfying that $|N_G[u]|< L$ from the facility set $\calF$, since $u$ can not be open in any optimal feasible solution.
\footnote{
If this deletion causes the induced graph unconnected, similar to Lemma 6 in \citep{kociumaka2014constant}, we divide the graph into different connected components, and consider each smaller induced instance based on different connected components.
}
Otherwise if $L\leq |N_G[u]|< U_u$, we set $U_u\leftarrow\min\{U_u,|N_G[u]|\}$, which has no influence on any optimal feasible solution of the induced \supnonuniout\ instance. The following lemma gives a useful property of \CCT.

\begin{lemma}
\label{lm:xi}
Given an induced \supnonuniout\ instance induced by an undirected unweighted connect graph $G=(\calC\cup \calF,E)$, and a core-center tree $T=(\calF,E_T)$, suppose $I$ is the core-center set of $T$. Then, we can construct a function $\xi:\calC\rightarrow \calF$ satisfying the following properties in polynomial time.
\begin{enumerate}
\item For all $v\in\calC$, we have $(\xi(v), v)\in E$;
\item For all $ u\in I$, we have $|\xi^{-1}(u)|\geq L$.
\end{enumerate}
\end{lemma}
\begin{proof}
 Firstly, for each pair $u\in I$ and $v\in N_G[u]$, we define $\xi(v)=u$. We can make this mapping since for each pair $u_1,u_2\in I$, we have $N_G[u_1]\cap N_G[u_2]=\emptyset$ by Definition \ref{def:CCT}. For the rest clients $v\in\calC$, we define $\xi(v)$ to be an arbitrary facility $u\in\calF$ adjacent to $v$.

By the above construction, the first constraint is satisfied naturally. The second constraint is satisfied by the fact that $|\xi^{-1}(u)|\geq |N_G[u]|\geq L$ for all $u\in I$.
\end{proof}
\subsection{A Simple Case: \supunisoftout}
\label{sec:simple}

We first consider a simple case where the capacity bounds (upper and lower) are uniform and soft.
In this setting, we want to find an open facility set $F=\{u_i\mid u_i\in \calF\}_i$. Note that we allow multiple open centers in $F$. We also need to find an assignment function $\phi:\calC\rightarrow F$, representing that we assign every client $v\in \calC$ to facility $\phi(v)$.
The main theorem is as follows.

\begin{theorem}
\label{thm:5app}
(main theorem)
There exists a 5-approximation polynomial time algorithm for the \supunisoftout\ problem.
\end{theorem}

By Lemma \ref{lm:bipartite}, we only consider induced \supunisoftout\ instances. Given an induced \supunisoftout\ instance induced by an undirected unweighted connect graph $G=(\calC\cup \calF,E)$, recall that we can assume $|N_G[u]|\geq U_u\geq L$ for each $u\in\calF$. We first construct a \CCT\ $T=(\calF, E_T)$ rooted at node $u^*$, and a function $\xi:\calC\rightarrow \calF$ satisfying Lemma \ref{lm:xi}. For a facility set $P\subseteq \calF$, we denote $\xi^{-1}(P)=\bigcup_{u\in P}\xi^{-1}(u)$ to be the collection of clients assigning to some facility in $P$ by $\xi$.

Our algorithm mainly includes two procedures. The first procedure is called \emph{pass-up}, which is a greedy algorithm to map clients to facilities from the leaves of $T$ to the root. After the 'pass-up' procedure, we still leave some unassigned clients nearby the root. Then we use a procedure called \emph{pass-down} to allocate those unassigned clients by iteratively finding an \emph{exchange route}. In the following, we give the details of both procedures.

\topic{Procedure Pass-Up.} Assume that $|\calC|=aL+b$ for some $a\in \bbbn$ and $0\leq b\leq L-1$.
In this procedure, we will find an open facility set $F$ of size $a$. We also find an assignment function $\phi$ which assigns $aL$ clients to some nearby facility in $F$ except a client set $S\subseteq \calC$. Here, $S$ is a collection of $b$ clients in $\xi^{-1}(u^*)$ nearby the root $u^*$. Our main idea is to open facility centers from the leaves of \CCT\ $T$ to the root iteratively. During opening centers, we assign exactly $L$ 'close' clients to each center. This is the reason that there are $b$ unassigned clients after the whole procedure.

We then describe an iteration of pass-up. Assume that $I$ is the core-center set of $T$. At the beginning, we find a non-leaf vertex $u\in I$ satisfying that all of its grandchildren (if exists) are leaves. We denote $P\subseteq \calF$ to be the collection of all children and all grandchildren of $u$. In the next step, we  consider all unscanned clients in $\xi^{-1}(P)$,
\footnote{Here, unscanned clients are those clients that have not been assigned by $\phi$ before this iteration.}
and assign them to the facility $u$. We want that each center at $u$ serves exactly $L$ centers. However, there may exist one center at $u$ serving less than $L$ unscanned clients in $\xi^{-1}(P)$. We assign some clients in $\xi^{-1}(u)$ to this center such that it also serves exactly $L$ clients. After this iteration, we delete the subtree rooted at $u$ from $T$ except $u$ itself. 

Finally, the root $u^*$ will become the only remaining node in $T$. We open multiple centers at $u^*$, each serving exactly $L$ clients in $\xi^{-1}(u^*)$, until there are less than $L$ unassigned clients. See Algorithm \ref{alg:passup} for details. We have the following lemma.


\begin{algorithm}
  \caption{Pass-Up}
\label{alg:passup}
    \textbf{Input:} an induced \supunisoftout\ instance induced by $G=(\calC\cup \calF,E)$, a \CCT\ $T=(\calF,E_T)$, and a function $\xi:\calC\rightarrow \calF$; \\
    Initialize $S\leftarrow \calC$, $T'\leftarrow T$, $j\leftarrow 0$;\\
    \While{$u^*\in T'$}{
    If the root $u^*$ is the only node of $T'$, we let $u\leftarrow u^*$. Otherwise, arbitrarily pick a non-leaf vertex $u\in I$ in $T'$ whose all grandchildren (if exists) are leaves of $T'$ ;\\
        Denote the subtree of $T'$ rooted at $u$ by $\hat{T}$. Denote $P$ to be the collection of all facilities in $\hat{T}\setminus \{u\}$; \\
        Let $l\leftarrow |\xi^{-1}(P)\cap S|$. Assume that $l=t L+q$ for some $t\in \bbbn$ and $0\leq q\leq L-1$; \\
        Arbitrarily pick $L-q$ clients from $\xi^{-1}(u)$ to form a set $K$. Let $A\leftarrow (\xi^{-1}(P)\cap S)\cup K$;\\
        Let $u_{j+1}\leftarrow u,u_{j+2}\leftarrow u,\cdots,u_{j+t+1}\leftarrow u$;\\
        For each center $u_{j+i}$ ($1\leq i\leq t+1$), assign exactly $L$ clients $v\in A$ to $u_{j+i}$, i.e., let $\phi(v)=u_{j+i}$;\\
        Let $j\leftarrow j+t+1$, $S\leftarrow S\setminus A$, $T'\leftarrow (T'\setminus \hat{T})\cup \{u\}$;\\
        \If{$u=u^*$}{
        Let $l'=|\xi^{-1}(u)\cap S|$. Assume that $l'=t'L+q'$ for some $t\in \bbbn$ and $0\leq q\leq L-1$; \\
        Arbitrarily pick $t'L$ clients from $\xi^{-1}(u)\cap S$ to form a set $K'$; \\
        Let $u_{j+1}\leftarrow u,u_{j+2}\leftarrow u,\cdots,u_{j+t'}\leftarrow u$; \\
        For each center $u_{j+i}$ ($1\leq i\leq t'$), assign exactly $L$ clients $v\in K'$ to $u_{j+i}$, i.e., let $\phi(v)=u_{j+i}$; \\
        Let $j\leftarrow j+t'$, $S\leftarrow S\setminus K'$, $T'\leftarrow \emptyset$;\\
        }
    }
    \textbf{Output:} $F=\{u_1,u_2,\cdots,u_j\}$, $\phi:(\calC\setminus S) \rightarrow F$ and $S$.\\
\end{algorithm}

\begin{lemma}
\label{lm:passup}
Given an induced \supunisoftout\ instance induced by an undirected unweighted connect graph $G=(\calC\cup \calF,E)$, assume that $|\calC|=aL+b$ for some $a\in \bbbn$ and $0\leq b\leq L-1$. The output of Algorithm \ref{alg:passup} satisfies the following properties:
\begin{enumerate}
\item Each open facility $u_j\in F$ satisfies that $u_j\in I$, and $|F|=a$;
\item The unassigned client set $S\subseteq \xi^{-1}(u^*)$, and $|S|=b$;
\item For each facility $u_i\in F$, we have $|\phi^{-1}(u_i)|=L$.
\item For each client $v\in \calC\setminus S$, $\phi(v)$ is either $\xi(v)$, or the parent of $\xi(v)$ in $T$, or the grandparent of $\xi(v)$ in $T$. Moreover, we have $\dist_G(v,\phi(v))\leq 5$.
\end{enumerate}
\end{lemma}
\begin{proof}
We first prove the feasibility of Algorithm \ref{alg:passup}. The feasibility of Line 7 follows from the fact that $|\xi^{-1}(u)|\geq L$ by Lemma \ref{lm:xi}. Since $0\leq q\leq L-1$ by Line 6, we can always pick $L-q$ clients from $\xi^{-1}(u)$. The feasibility of Line 9 follows from the fact that $|A|=l+L-q=(t+1)L$. Since we open $(t+1)$ centers at $u$, it is able to assign exactly $L$ clients in $A$ to each center.

Then we prove the properties of the output. The first three properties mainly follow from the fact that $|\calC|=aL+b$ and we assign exactly $L$ clients to each open center. We only need to verify that $S\subseteq \xi^{-1}(u^*)$. By Line 3, we always pick $u=u^*$ in the last iteration of Algorithm \ref{alg:passup}. By Line 11-16, this fact is obvious. For each center at $u$, it only serves $L$ clients in $(\xi^{-1}(P)\cap S)\cup K$. By the definition of $P$ and $K$, we conclude the first part of the last property.
Moreover, we have $\dist_G(v,\xi(v))=1$ by Lemma \ref{lm:xi} and $\dist_G(\xi(v),u_{\phi(v)})\leq 4$ by Definition \ref{def:CCT}. By the triangle inequality, we have $\dist_G(v,\phi(v))\leq 5$.
\end{proof}

\topic{Procedure Pass-Down.}
After the procedure \emph{pass-up}, we still leave an unassigned client set $S$ of size $b$. However, our goal is to serve at least $p$ clients. Therefore, we need to modify the assignment function $\phi$  and serve more clients.

The procedure pass-down handles the remaining $b$ clients in $S$ one by one, see Algorithm \ref{alg:passdown} for details. At the beginning of pass-down, we initialize an 'unscanned' client set $B\leftarrow \calC\setminus S$, i.e., $B$ is the collection of those clients allowing to be reassigned by pass-down. In each iteration, we arbitrarily pick a client $v\in S$ and assign it to the root node $u^*$. However, if each open facility at $u^*$ has already served $U_{u^*}$ clients by $\phi$, assigning $v$ to $u^*$ will violate the capacity upper bound. In this case, we actually find an open center $u_j\in F$ such that $|\phi^{-1}(u_j)|<U_j$, i.e., there are less than $U_j$ clients assigned to $u_j$ by $\phi$. We then construct an \emph{exchange route} consisting of open facilities in $F$. We first find a sequence of nodes $w_0=u^*,w_1,\cdots,w_m=u_j$ in $T$ satisfying that $w_i$ is the grandparent of $w_{i+1}$ in the core-center tree $T$ for all $0\leq i\leq m-1$. Then for each node $w_i$ $(1\leq i\leq m-1)$, we pick a client $v_i\in \xi^{-1}(w_i)$ which has not been reassigned so far. We call such a sequence of clients $v, v_1,\ldots,v_{m-1}$ an exchange route. Our algorithm is as follows: 1) we assign $v$ to $\phi(v_1)$; 2) we iteratively reassign $v_i$ to $\phi(v_{i+1})$ in order $(1\leq i\leq m-2)$; 3) finally we reassign $v_{m-1}$ to $u_j$. We then mark all clients $v_i$ $(1\leq i\leq m-1)$ in the exchange route by removing them from the 'unscanned' client set $B$. Note that our exchange route only increases the number of clients assigned to $u_j$ by one.
We will prove such an exchange route always exists in each iteration.
Thus in each iteration, the procedure pass-down assigns one more client $v\in S$ to some open facility in $F$. We will argue that there are at least $p$ clients served by $F$ at the end of pass-down.


\begin{algorithm}
  \caption{Pass-Down}
  \label{alg:passdown}
  \textbf{Input:} an induced \supunisoftout\ instance induced by $G=(\calC\cup \calF,E)$, a \CCT\ $T=(\calF,E_T)$, a function $\xi:\calC\rightarrow \calF$, an open facility set $F=\{u_1,u_2,\cdots,u_a\}$, an unassigned client set $S$, and a function $\phi:(\calC\setminus S) \rightarrow F$; \\
    Initialize $B\leftarrow \calC\backslash S$;\\
    \While{$S\neq \emptyset$ and $\exists 1\leq j\leq a$, $|\phi^{-1}(j)|<U_j$}{
        Arbitrarily pick a client $v\in S$ and an open facility $u_j$ ($1\leq j\leq a$) satisfying that $|\phi^{-1}(u_j)|<U_j$;\\
        \eIf{$u_j=u^*$}{
            Let $\phi(v)\leftarrow u_j$, $S\leftarrow S\setminus \{v\}$;\\
        }{
            Let $w_0=u^*,w_1,\cdots,w_m=u_j$ be the sequence of nodes in $T$ where $w_i$ is the grandparent of $w_{i+1}$ for all $1\leq i\leq m-1$. \\
            Let $v_0\leftarrow v$. For every $1\leq i\leq m-1$, arbitrarily pick a client $v_i\in\xi^{-1}(w_i)\cap B$;\\
            \For{$i=0,\cdots,m-2$}{
                Reassign $\phi(v_i)\leftarrow \phi(v_{i+1})$;\\
            }
            Reassign $\phi(v_{m-1})\leftarrow u_j$;\\
            Let $S\leftarrow S\setminus \{v\}$, $B\leftarrow B\setminus\{v_1,v_2,\cdots,v_{m-1}\}$;\\
%
            Let $\phi(v_{m-1})\leftarrow u_j$, $S\leftarrow S\setminus \{v\}$;\\
        }
    }
    \textbf{Output:} $C\leftarrow \calC\setminus S$, $F=\{u_1,u_2,\cdots,u_a\}$ and $\phi: C\rightarrow F$;

\end{algorithm}


Now we prove the following lemma. Note that Theorem \ref{thm:5app} can be directly obtained by Lemma \ref{lm:bipartite} and Lemma \ref{lm:passdown}.

\begin{lemma}
\label{lm:passdown}
Algorithm \ref{alg:passdown} outputs a distance-5 solution $(C,F,\phi)$ of the given induced \supunisoftout\ instance induced by $G=(\calC\cup \calF,E)$ in polynomial time.

\end{lemma}

\begin{proof}
We first verify the feasibility of Algorithm \ref{alg:passdown}. The feasibility of Line 8 follows from the fact that $u^*, u_j\in I$ by Lemma \ref{lm:passup}. Then we only need to show that an exchange route $\{v_0,\cdots, v_{m-1}\}$ described in Line 9-14 must exist in each iteration. On one hand, we verify that an exchange route $v,v_1,\ldots,v_{m-1}$ in Line 9 always exists. Since we assign one more client in $S$ in each iteration. Thus, there are at most $b$ iterations by Property 2 in Lemma \ref{lm:passup}. Then the node $w_t\in I$ appears at most $b$ times in the sequence in Line 8, and at most $b<L$ clients in $\xi^{-1}(w_t)$ are removed from $B$ at the end of the algorithm. By Lemma \ref{lm:xi}, we also have $|\xi^{-1}(w_t)|\geq L$. Thus we always have $\xi^{-1}(w_t)\cap B\neq \emptyset$ in Line 12, which proves the existence of an exchange route. On the other hand, for each open facility $\phi(v_i)$ ($1\leq i\leq m-1$), the number of clients served by $\phi(v_i)$ dose not change after Line 12. It is because we reassign $v_{i-1}$ to $\phi(v_i)$, and remove $v_i$ from it.

We then show that for all $v\in\calC$, we have $\dist_G(v,\phi(v))\leq 5$ at the end of the algorithm. In Line 11, we reassign a client $v_{t-1}\in \xi^{-1}(w_{t-1})$ to $\phi(v_t)$. Since $w_t$ is the grandchild of $w_{t-1}$, we have $\dist_G(v_t, w_{t-1})\leq 4$ by Property 4 in Lemma \ref{lm:passup}. Combining with $v_{t-1}\in \xi^{-1}(w_{t-1})$, we conclude that $\dist_G(v_{t-1},\phi(w_t))\leq 5$. Then by Property 4 in Lemma \ref{lm:passup}, we finish the proof.

Finally, we show that the output client set $C$ is of size at least $p$. In fact, we only need to prove that $\sum_{j=1}^a U\geq p$, i.e., $aU\geq p$. Since the number of open facility centers in the optimal solution served at least $p$ clients is at most $a=\lfloor |\calC|/L \rfloor$, we have $p\leq aU$.
\end{proof}

\topic{\uniout.} Consider the \uniout\ problem with hard capacities. We first treat a given induced \uniout\ instance as an induced \unisoftout\ instance. Then we apply Theorem \ref{thm:5app} and obtain a 5-approximation solution $(C,F,\phi)$. Since the two instances are induced by the same connected graph, the optimal capacitated center value of the induced \unisoftout\ instance is at most the optimal capacitated center value 1 of the induced \uniout\ instance. Therefore, we know $\max_{v\in C}\dist(v,\phi(v))\leq 5$. Since we have hard capacities, we still need to modify $F$ to be a single set. In fact, we can choose arbitrary vertex $v_i\in \phi^{-1}(u_i)$ to replace each $u_i\in F$ as an open center, and assign all vertices in $\phi^{-1}(u_i)$ to $v_i$. Note that the distance between any $v\in \phi^{-1}(u_i)$ and $v_i$ is at most  $5$ and the new $u_i$ is at most $10$. Thus we have the following theorem.

\begin{theorem}
\label{thm:uniout}
There exists a 10-approximation polynomial time algorithm for the \uniout\ problem.
\end{theorem}


\subsection{\nonuniout}
\label{sec:generalized}

In this subsection, we consider a more complicated case where the capacity upper bounds are non-uniform, and each vertex has a hard capacity. Our main theorem is as follows.

\begin{theorem}
\label{thm:11app}
(main theorem)
There exists an 11-approximation polynomial time algorithm for the \nonuniout\ problem.
\end{theorem}

By Lemma \ref{lm:bipartite}, we only need to consider induced \supnonuniout\ instances. For an induced \supnonuniout\ instance induced by an undirected unweighted connected graph $G=(V=\calC\cup \calF,E)$, recall that we can assume $U_u\leq |N_G(u)|$ for every vertex $u\in \calF$.
\footnote{Recall that we may remove some facilities from $\calF$ such that this assumption is satisfied. Thus, the set $\calF$ may be a subset of $V$.}
Since we consider the center version, every vertex $v\in\calC$ has an individual capacity interval $[L,U_v]$ and can be opened as a center as well. This fact is useful for our following algorithm and is the reason why we do not consider the supplier version in this subsection.

Similar to \unisoftout, our algorithm first computes a core-center tree $T=(\calF,E)$ rooted at $u^*$, a core-center set $I$ and a function $\xi$ described as in Lemma \ref{lm:xi}. Assume that $|\calC|=aL+b$ for some $a\in \bbbn$ and $0\leq b\leq L-1$. Note that the procedure \emph{pass-up} algorithm does not depend on the capacity upper bounds. Therefore, we still use the procedure \emph{pass-up} to compute an open set $F=\{u_1,u_2,\cdots,u_a\}$, an unassigned set $S\subseteq \xi^{-1}(u^*)$ of size $b<L$, and a function $\phi:(\calC\setminus S) \rightarrow F$.

However, we can not apply \emph{pass-down} directly. On one hand, since we consider non-uniform capacity upper bounds, the inequality $\sum_{j=1}^a U_{u_j}\geq p$ may not be satisfied. We need to choose open centers carefully such that at least $p$ vertices can be served. On the other hand, we can not open multiple facilities in a single vertex by hard capacities. Thus, we need the following lemma to modify the open center set $F$.

\begin{lemma}
\label{lm:matching}
Given an induced \nonuniout\ instance induced by $G=(V=\calC\cup \calF,E)$ where $|N_G(u)|\geq U_u$ for each $u\in \calF$ and an open set $F=\{u_1,u_2,\cdots,u_a\}$ computed by \emph{pass-up}, there exists a polynomial time algorithm that finds another open set $F'=\{u'_1,u'_2,\ldots,u'_a\}$ satisfying the following properties:
\begin{enumerate}
\item $F'$ is a single set.
\item For all $1\leq i\leq a$, we have $\dist_G(u_1,u'_1)\leq 6$.
\item $\sum_{i=1}^{a}U_{u'_i}\geq p$.
\end{enumerate}
\end{lemma}

We will prove the above lemma later. By Lemma \ref{lm:matching}, we are ready to prove Theorem \ref{thm:11app}.
\vspace{0.3cm}

\noindent \emph{Proof of Theorem \ref{thm:11app}.}
By Lemma \ref{lm:matching}, we obtain another open set $F'=\{u'_1,u'_2,\ldots,u'_a\}$. We first modify $U_{u_i}$ to be $U_{u'_i}$ for all $1\leq i\leq a$. Then we apply the procedure \emph{pass-down} according to the modified capacities. By Lemma \ref{lm:passdown}, we obtain a distance-5 solution $(C,F,\phi)$. Since $\sum_{i=1}^{a}U_{u'_i}\geq p$, at least $p$ vertices are served by $\phi$. Finally, for each vertex $v\in C$ and $u_i\in F$ such that $\phi(v)=u_i$, we reassign $v$ to $u'_i\in F'$, i.e., let $\phi(v)=u'_i$. By Lemma \ref{lm:matching}, we obtain a feasible solution for the given induced \nonuniout\ instance. Since $\dist(u_i,u'_i)\leq 6$ $(1\leq i\leq a)$, the capacitated center value of our solution is at most $5+6=11$. Combining with Lemma \ref{lm:bipartite}, we finish the proof.

\vspace{0.3cm}
Now we only need to prove Lemma \ref{lm:matching}.

\begin{proof}
We construct an undirected weighted bipartite graph $B=(F,\calC;E_B)$ as follows: $(u_j,v)\in E_B$ ($u_j\in F,v\in \calC$) if and only if $\dist(u_j,v)\leq 6$ and this edge has weight $U_v$. We then find a maximum-weight maximum-matching $M$ on this graph $B$. We only need to verify that $F$ is perfectly matched in $M$ and the total weight of $M$ is at least $p$. Suppose each $u_i\in \calF$ is matched to $v_i\in \calC$, we finish the proof by letting $u'_i=v_i$.

Define $O\subseteq \calF\subseteq \calC$ to be the optimal open center set. By Hall's theorem, we first prove the existence of a matching $M_1$ in $B$ satisfying that every vertex in $O$ is matched. For any subset $O'\subseteq O$, we assume by contradiction that $|N_B(O')|\leq |O'|-1$, where $N_B(O')=\cup_{v\in O'}N_B(v)$. Define $C'$ to be the set of vertices served by $O'$ in the optimal solution. By the capacity lower bound, we have $|C'|\geq L|O'|$. Recall that $S$ is the unassigned set of size $b<L$ obtained by pass-up. Therefore, we have $|C'\setminus S|>L(|O'|-1)\geq L|N_B(O')|$. On the other hand, for each vertex $v\in C'\setminus S$, we have $\dist_G(v,\phi(v))\leq 5$ by Lemma \ref{lm:passup}. Thus, we conclude that $\dist_G(\phi(v),O')=\min_{w\in O'}\dist_G(\phi(v),w)\leq 6$, which implies all $\phi(v)\in N_B(O')$. Since each open center $u_i\in F$ only serves $L$ vertices by Lemma \ref{lm:passup}, we have $|C'\setminus S|\leq L|N_B(O')|$ which is a contradiction. So we prove the existence of $M_1$. Note that the total weight of $M_1$ is at least $p$.

Note that there exists a matching on $B$ such that $F$ is perfectly matched. We can achieve this property by matching each $u_i\in \calF$ to an arbitrary vertex $v_i$ such that $\phi(v_i)=u_i$. Then by Hungarian Algorithm, we can construct a matching $M_2$ by iteratively finding augmenting paths based on $M_1$, until all vertices in $F$ are matched. Since any augmenting path can not make a matching vertex unmatched, we conclude that the total weight of $M_2$ is at least $p$. Thus, the maximum-weight maximum-matching $M$ on $B$ must satisfy that $F$ is perfectly matched and the total weight is at least $p$..
\end{proof}

\topic{\supnonunisoftout.} Consider the \supnonunisoftout\ problem with soft capacities.  Our technique is similar to \nonuniout\ except a difference procedure for choosing $F'$ in Lemma \ref{lm:matching}.
By Lemma \ref{lm:bipartite}, we again consider a given induced \supnonunisoftout\ instance induced by an undirected unweighted connected $G=(\calC\cup\calF,E)$. W.l.o.g., we assume that $|N_G(u)|\geq U_u$ for all facilities $u\in \calF$. Similarly, we compute $T,I,\xi$ and apply \emph{pass-up} to compute $F=\{u_1,u_2,\ldots,u_a\}$, $S\subseteq \xi^{-1}(u^*)$ and $\phi$. Before applying \emph{pass-down}, we also need to find another open facility set $F'=\{u'_1,\ldots,u'_a\}$. Since we have soft capacities, we only need to require $F'$ to satisfy Property 2 and 3 in Lemma \ref{lm:matching}. This is the reason why we can consider the supplier version. By the same technique as in the proof of Theorem \ref{thm:11app}, we have the following theorem.

\begin{theorem}
\label{lm:soft11}
There exists a poly-time algorithm achieving approximation ratio 11 for \supnonunisoftout\ problem.
\end{theorem}
\begin{proof}

We only need to find another open facility set $F'$ satisfying Property 2 and 3 in Lemma \ref{lm:matching}. In fact, we simply define $u'_i=\arg\max_{u\in \calF:\dist_G(u,u_i)\leq 6} U_u$ for all $u_i\in F$. We only need to verify that $\sum_{i=1}^a U_{u'_i}\geq p$. Assume that the optimal open facility set is $O=\{u_1^*, u_2^*,\cdots, u^*_m\}$ and the optimal assignment function is $\phi^*$. W.l.o.g., we assume that $U_{u_1^*}\geq U_{u_2^*}\geq\cdots\geq U_{u_m^*}$. We only need to find an injection $\sigma:\{1,2,\cdots,w\}\rightarrow\{1,2,\cdots,a\}$ such that $U_{u_j^*}\leq U_{u'_{\sigma(j)}}$.

The injection can be found greedily.
Suppose $\sigma(1),\sigma(2),\cdots,\sigma(j-1)$ have been decided for some $1\leq j\leq w$, we want to decide $\sigma(j)$. Since each $u_i\in F$ only serves $L$ clients by $\phi$ and the unassigned client set $S$ is of size $b<L$, we have the following property by counting:  $$A:=\left(\bigcup_{i=1}^j\phi^{*-1}(u^*_i)\right)\setminus\left(S\cup\bigcup_{i=1}^{j-1}\phi^{-1}(u_{\sigma(i)})\right)\neq \emptyset.$$
Arbitrarily pick a client $v\in A$. Define $\sigma(j)=\phi(v)$. By the definition of $O$, there exists some $1\leq j'\leq j$ such that $\dist_G(v,u_{j'}^*)\leq 1$. Therefore, $\dist_G(u_{\sigma(j)},u_{j'}^*)\leq 6$. Thus, we have $U_{u_j^*}\leq U_{u_{j'}^*}\leq U_{u'_{\sigma(j)}}$ by the definition of $u'_{\sigma(j)}$. The proof is complete.
\end{proof}

\section{Capacitated $k$-Center with Two-Sided Bounds and Outliers}
\label{sec:algtree}

In this section, we study the capacitated $k$-center problems with two-sided bounds, with or without outliers, and give approximation algorithms. We consider the case that all vertices have a uniform capacity lower bound $L_v=L$, while the capacity upper bounds can be either uniform or non-uniform. Our goal is to propose approximation algorithms with constant approximation ratio.
Similar to \citep{an2015centrality,kociumaka2014constant}, we use the standard LP relaxation and the rounding procedure distance-$r$ transfer. We will first extend the distance-$r$ transfer procedure for two-sided bounds.

\subsection{LP Formulation}
We first give a natural LP relaxation for \ksupbnonuniout.
\begin{definition}

\label{def:disrelax}
(Distance-$r$ relaxation $\mathsf{LP}_{r}(G)$) Given an \ksupbnonuniout\ instance, the following feasibility $\mathsf{LP}_{r}(G)$ that fractionally verifies whether there exists a solution that assigns at least $p$ clients to an open center of distance at most $r$:
\begin{center}
\fbox{
$\begin{array}{ll}
0\leq x_{uv},y_u\leq 1,&\forall u\in\calF,v\in\calC;  \\
x_{uv}=0,&\mathrm{if\ } \dist(u,v)>r;\\
  x_{uv}\leq y_u,&\forall u\in \calF, v\in \calC; \\
\sum_{u\in \calF}y_u=k;&\\
\sum_{u\in\calF, v\in\calC}x_{uv}\geq p;&\\
\sum_{u\in \calF}x_{uv}\leq 1,&\forall v\in \calC;\\
 L_uy_u\leq\sum_{v\in \calC}x_{uv}\leq U_uy_u,&\forall u\in \calF.
\end{array}$}
\end{center}
Here $x_{uv}$ is called an assignment variable representing the fractional amount of assignment from client $v$ to center $u$, and $y_u$ is called the opening variable of $u\in \calF$. For convenience, we use $x,y$ to represent $\{x_{uv}\}_{u\in\calF,v\in \calC}$ and $\{y_u\}_{u\in\calF}$, respectively.
\end{definition}

By Definition \ref{def:bipartite}, $\mathsf{LP}_1(G)$ must have a feasible solution for any induced\ksupbnonuniout\ instance. Assume that we have a feasible fractional solution $(x,y)$ of $\mathsf{LP}_1(G)$. We want to obtain a distance-$\rho$ solution by rounding $(x,y)$. We then recall a rounding procedure called distance-$r$ transfer.

\subsection{Distance-r Transfer}
\label{sec:rounding}
We first extend the definition of distance-$r$ transfer proposed in \citep{an2015centrality,kociumaka2014constant} by adding the third condition. For a vertex $a\in \calC\cup \calF$ and a set $B\subseteq \calC\cup \calF$, we define $\dist(a,B)=\min_{b\in B}\dist(a,b)$.
\begin{definition}
\label{def:transfer}
Given an \ksupbnonuniout\ instance and $y\in\bbbr^{\calF}_{\geq 0}$, a vector $y'\in\bbbr^{\calF}_{\geq 0}$ is a distance-$r$ transfer of $y$ if
\begin{enumerate}
\item $\sum_{u\in \calF}y'_u=\sum_{u\in \calF}y_u$;
\item $\sum_{w\in \calF: \dist(w,W)\leq r}U_w y'_w\geq\sum_{u\in W}U_u y_u$ for all $W\subseteq \calF$;
\item $\sum_{w\in \calF: \dist(w,W)\leq r}L_w y_w\geq\sum_{u\in W}L_uy_u'$ for all $W\subseteq \calF$.
\end{enumerate}
If $y'$ is a characteristic vector of $F\subseteq \calF$, we say that $F$ is an integral distance-$r$ transfer of $y$.
\end{definition}

Recall that the first condition says that a transfer should not change the total number of open centers. By an argument using Hall's theorem as in \citep{an2015centrality,kociumaka2014constant}, the second condition is important for satisfying the capacity upper bounds. In this paper, we add the third condition to satisfy the capacity lower bounds. Like in \citep{an2015centrality,kociumaka2014constant}, an integral distance-$r$ transfer of the fractional solution of $\mathsf{LP}_{r}(G)$ already gives a distance-$(r+1)$ solution by the following lemma.

\begin{lemma}
\label{lm:r+1}
Given an \ksupbnonuniout\ problem, assume $(x,y)$ is a feasible solution of $\mathsf{LP}_1(G)$ and $F\subseteq \calF$ is an integral distance-$r$ transfer of $y$. Then one can find a distance-$(r+1)$ solution $(C,F,\phi)$ in polynomial time.
\end{lemma}

\begin{proof}
Consider a bipartite graph $H=(\calC,\calF,E_H)$ with $(v,u)\in E_H$ ($v\in \calC, u\in \calF$) if $\dist_G(v,u)\leq r+1$. Modify $H$ to obtain $H_1=(\calC,F_1,E_{H_1})$ by removing vertices from $\calF\setminus F$ and duplicating each vertex $u\in F$ to its capacity lower bound, i.e. $L_u$ times. Then we show that there exists a matching $M_1$ such that every vertex in $F_1$ is matched. By Hall's theorem, we need to prove that $|\{v\in\calC:\dist_{H'}(v,W)\leq r+1\}|\geq |W|$ for any $W\subseteq F_1$.
By the construction of $F_1$, we in fact only need to prove that for each $W\subseteq F$, we have the following:
$$|\{v\in\calC:\dist_G(v,W)\leq r+1\}|\geq \sum_{u\in W}L_u.$$
By the third condition of Definition \ref{def:transfer}, we know $\sum_{u\in W}L_u\leq\sum_{w\in\calF: \dist_G(w,W)\leq r}L_wy_w$. Then by the LP constraint, we have the following inequality:

\begin{eqnarray*}
& &\sum_{u\in W}L_u\leq\sum_{w\in\calF: \dist_G(w,W)\leq r}L_wy_w\leq \sum_{w\in\calF:\dist_G(w,W)\leq r}\sum_{v\in\calC: \dist_G(w,v)\leq 1}x_{wv} \\
&\leq&\sum_{v\in\calC: \dist_G(v,W)\leq r+1}\sum_{w\in\calF}x_{wv}\leq \sum_{v\in\calC: \dist_G(v,W)\leq r+1}\sum_{w\in\calF}1=|\{v\in\calC:\dist_G(v,W)\leq r+1\}|.
\end{eqnarray*}

We then modify $H$ to obtain another bipartite graph $H_2=(\calC,F_2,E_{H_2})$ by removing vertices from $\calF\setminus F$ and duplicating each vertex $u\in F$ to its capacity upper bound, i.e. $U_u$ times. Since $L_u\leq U_u$ for all $u\in \calF$, we can consider $F_1$ as a subset of $F_2$. Moreover, we can consider $M_1$ as a matching of $H_2$ satisfying that for each $u\in F$, there are exactly $L_u$ duplicates of $u$ to be matched. On the other hand, there exists a matching $M_2$ of $H_2$ satisfying that there are at least $p$ clients $v\in \calC$ that are matched \citep[Lemma 9]{kociumaka2014constant}. Finally, we show that it is possible to modify $M_1$ to be a matching with the same property as $M_2$. Starting with $M_1$, we iteratively find an augmenting path on $H_2$ such that there is one more client that are matched. We stop until there are at least $p$ matched clients. The feasibility is not hard by Hungarian Algorithm. Since any augmenting path can not make a matched vertex unmatched, we conclude that for each $u\in F$ there are at least $L_u$ matched duplicates. This matching can be found in polynomial time, and $|F|=k$ follows from the first condition of Definition \ref{def:transfer}.
\end{proof}

Note that the capacity lower bound is uniform, we have the following lemma for distance-$r$ transfer.

\begin{lemma}
\label{lm:luni}
Given an \ksupnonuniout\ instance and $y\in\bbbr^{\calF}_{\geq 0}$, assume that a vector $y'\in\bbbr^{\calF}_{\geq 0}$ satisfies the following conditions:
\begin{enumerate}
\item $\sum_{u\in \calF}y'_u=\sum_{u\in \calF}y_u$;
\item $\sum_{w\in \calF: \dist(w,W)\leq r}U_w y'_w\geq\sum_{u\in W}U_u y_u$ for all $W\subseteq \calF$;
\item There exists a function $g:\calF\times \calF\rightarrow \R_{\geq 0}$ satisfying the following conditions:
    \begin{enumerate}
    \item for all $u,w\in \calF$, if $\dist(u,w)>r$, we have $g(u,w)=0$;
    \item for all $u\in \calF, y_u=\sum_{w\in \calF}g(u,w),y'_u=\sum_{w\in \calF}g(w,u)$.
    \end{enumerate}
\end{enumerate}
Then $y'$ is a distance-$r$ transfer of $y$. We call this $y'$ a local distance-$r$ transfer of $y$.
\end{lemma}

\begin{proof}
By Definition \ref{def:transfer}, we only need to verify that $\sum_{w\in \calF: \dist(w,W)\leq r}L_w y_w\geq\sum_{u\in W}L_uy_u'$ for all $W\subseteq \calF$. Since the capacity lower bound is uniform, we can simplify the above inequality as follows: $\sum_{w\in \calF: \dist(w,W)\leq r}y_w\geq\sum_{u\in W}y_u'$. We finish the proof by the following argument:
\begin{eqnarray*}
\sum_{u\in W}y_u'&=&\sum_{u\in W} \sum_{w\in \calF} g(w,u)=\sum_{u\in W} \sum_{w\in \calF:\dist(w,u)\leq r} g(w,u) \leq \sum_{u\in W} \sum_{w\in \calF:\dist(w,W)\leq r} g(w,u)\\
&=&\sum_{w\in \calF:\dist(w,W)\leq r}\sum_{u\in W}g(w,u)=\sum_{w\in \calF: \dist(w,W)\leq r}y_w.
\end{eqnarray*}
The second equality follows from condition 3a.
\end{proof}

The function $g$ in the above lemma can be considered as a flow from the facility set $\calF$ to its duplicate $\calF'$, with the property that no flow enters $w\in \calF'$ from any facility $u\in \calF$ such that $\dist(u,w)>r$. In other words, such a distance-$r$ transfer does not make a transfer of opening variables from some facility $u\in \calF$ to a 'too far' facility $w\in \calF$ satisfying that $\dist(u,w)>r$.

Similarly, if the capacity upper bound is uniform, we have the following lemma which may give some intuition for solving the ($\{L_u\}$,$U$,$k$,$p$)-\textsc{Supplier} problems in the future.

\begin{lemma}
\label{lm:uuni}
Given an ($\{L_u\}$,$U$,$k$,$p$)-\textsc{Supplier} instance and $y\in\bbbr^{\calF}_{\geq 0}$, assume that a vector $y'\in\bbbr^{\calF}_{\geq 0}$ satisfies the following conditions:
\begin{enumerate}
\item $\sum_{u\in \calF}y'_u=\sum_{u\in \calF}y_u$;
\item $\sum_{w\in \calF: \dist(w,W)\leq r}L_w y_w\geq\sum_{u\in W}L_uy_u'$ for all $W\subseteq \calF$
\item There exists a function $g:\calF\times \calF\rightarrow \R_{\geq 0}$ satisfying the following conditions:
    \begin{enumerate}
    \item $\forall u,w\in \calF$, if $\dist(u,w)>r$, we have $g(u,w)=0$;
    \item $\forall u\in \calF, y_u=\sum_{w\in \calF}g(u,w),y'_u=\sum_{w\in \calF}g(w,u)$.
    \end{enumerate}
\end{enumerate}
Then $y'$ is a distance-$r$ transfer of $y$.
\end{lemma}

\subsection{Capacitated $k$-Center with Two-Sided Bounds and Outliers}
\label{sec:kcenter}

Now we are ready to solve the \ksupnonuniout\ problem.
By Lemma \ref{lm:r+1} and \ref{lm:luni}, we only need to find an integral local distance-$r$ transfer. In \citep{an2015centrality,kociumaka2014constant}, there exists a polynomial time algorithm that finds an integral distance-$r$ transfer for the ($L$,$\{U_u\}$,$k$,$p$)-\textsc{Supplier} problem with capacity lower bound $L=0$. Their algorithm can be naturally generalized to the case that $L>0$. Thus, we only need to show that their rounding scheme also finds a local distance-$r$ transfer.
In this section, we briefly state the rounding schemes for the $k$-center problem with only capacity upper bounds in previous work \citep{an2015centrality,kociumaka2014constant}, and show that they find integral local distance-$r$ transfers. We first have the following observation.

\begin{observation}
\label{ob:local}
Given an \ksupnonuniout\ instance, let $y,y',y''\in\bbbr^{\calF}_{\geq 0}$. Assume $y'$ is a local distance-$r$ transfer of $y$ and $y''$ is a local distance-$r'$ transfer of $y'$. Then $y''$ is a local distance-$(r+r')$ transfer.
\end{observation}

Though the authors consider many variants in \citep{an2015centrality,kociumaka2014constant}, the rounding schemes are quite similar. We first take the (0,$\{U_u\}$,$k$)-\textsc{Center} problem as an example. \cite{an2015centrality} proposed a rounding scheme for this problem by reducing to a tree instance defined as follows.

\begin{definition} (Tree instance)
\label{def:tree}
A tree instance is defined as a tuple $(T,\{L_v\},\{U_v\},y)$, where $T=(V,E)$ is a rooted tree with a capacity interval $[L_v,U_v]$ for each facility $v\in V$ ($L_v,U_v\in \Z_{\geq 0}$ and $L_v\leq U_v$),
a distance function $\dist_T: V\times V\rightarrow\bbbr_{\geq 0}$ induced by $T$, and opening variables $y\in [0,1]^V$ satisfying that $\sum_{u\in V}y_v$ is an integer and $y_v=1$ for every non-leaf facility $v\in V$.
\end{definition}

The rounding procedure is as follows.
\begin{enumerate}
\item Find a fractional solution $(x,y)$ of the standard LP.
\item Construct a distance-1 transfer $y'$ of $y$ by aggregating one opening unit to each vertex in a set $V$ from its neighbors.
\footnote{Here, $V$ is a set of auxiliary vertices augmenting to the original instance. See \citep{an2015centrality} for more details.}
 Based on $y'$, construct a tree instance $(T=(V,E_T),0,\{U_u\},y')$.
\item Find an integral distance-2 transfer $y''$ of $y'$ based on the tree induced distance by \citep[Lemma 9]{an2015centrality}. Since the tree $T$ satisfies that the original distance between any two adjacent vertices is at most 3, $y''$ is also a distance-6 transfer of $y'$ based on the original distance.
\item Finally, find an open set $F\in \calF$ which is also a distance-1 transfer of $y''$. By the transitivity of distance-$r$ transfer, $F$ is a distance-8 transfer of $y$. By Lemma \ref{lm:r+1}, $F$ is a distance-9 solution.
\end{enumerate}

Our goal is to show that the above rounding scheme can be generalized to the case that $L>0$. We use exactly the same rounding scheme except that the initial fractional solution $(x,y)$ should satisfy all capacity lower bounds. We then find an open set $F$ as an integral distance-8 transfer of $y$, and we want to show that $F$ is also local. It is not hard to verify the local property is satisfied in Step 2 and Step 4. Thus, we only need to check Step 2 for the tree instance. In fact, we have the following lemma.

\begin{lemma}
\label{lm:treeinstance}
Given a tree instance $(T,L,\{U_u\},y)$, one can find in polynomial time an integral local distance-$2$ transfer of $y$.
\end{lemma}

\begin{proof} We recall the rounding scheme of tree instances in \citep{an2015centrality}.

If $T$ only consists of at most $1$ vertex, then $y$ is an integral vector itself, yielding an integral local distance-1 transfer of $y$.

If $T$ consists of more than $1$ vertex, we construct a distance-$2$ transfer by a recursion. We first find a non-leaf vertex $r\in T$ whose children are all leaves. We use $T_1=(V_1,E_1)$ to denote the subtree rooted at $r$, and use $T_2=(V_2,E_2)$ to denote the subtree $(T\setminus T_1)\cup \{r\}$ by deleting all children of $r$ from $T$. Define $y^1\in [0,1]^{V_1}$ and $y^2\in [0,1]^{V_2}$ as follows: $y^1|_{V_1\setminus \{r\}}=y|_{V_1\setminus \{r\}}$, $y^2|_{V_2\setminus\{r\}}=y|_{V_2\setminus\{r\}}$ and $y_r=y^1_r+y^2_r$ where $0\leq y^1_r<1$ and $\sum_{u\in V_1}y^1_u$ is an integer. Since $y_r=1$, we can always find such values $y^1_r$ and $y^2_r$. Define $Y=\sum_{u\in V_1}y^1_u$.
By the above definition, we observe that $Y$ is an integer in $[0,|V_1|-1]$. We use $F_1$ to denote the collection of $Y$ vertices among the children of $r$ of highest capacity upper bounds. We use $u^*$ to denote the child of $r$ of the $Y$-th highest capacity upper bound. Note that $u^*\in F_1$. We consider the following cases.
\begin{enumerate}
\item If $U_r<U_{u^*}$, then $F_1$ is an integral distance-$2$ transfer of $y_1$ by \citep[Lemma 9]{an2015centrality}. To prove that $F_1$ is also a local distance-$2$ transfer of $y_1$, we only need to verify the third condition in Lemma \ref{lm:luni}. Note that $|F_1|=Y$. There must exist a function $g:V_1\times V_1\rightarrow \R_{\geq 0}$ satisfying condition 3b. On the other hand, the diameter of the tree $T_1$ is at most 2. So the function $g$ must satisfy condition 3a, which finishes the proof. Then the algorithm recursively solves a smaller tree instance and obtains $F_2\subseteq V_2$ as an integral local distance-$2$ transfer of $y_2$. Since $F_1\cap F_2=\emptyset$, we conclude that $F:=F_1\cup F_2$ is an integral local distance-$2$ transfer of $y$.
\item If $U_r= U_{u^*}$, then both $F_1$ and $F'_1:=(F_1\cup\{r\})\setminus\{u^*\}$ are integral local distance-$2$ transfers of $y_1$ by the same argument as in the first case. In this case, we first recursively solves a smaller tree instance and obtains $F_2\subseteq V_2$ as an integral local distance-$2$ transfer of $y_2$. Then we define $$F:=\left\{\begin{array}{ll}F_1\cup F_2,&r\in F_2\\ F_1'\cup F_2,&r\notin F_2\end{array}\right.,$$
    which is an integral local distance-$2$ transfer of $y$.
\item If $U_r>U_{u^*}$, we modify the capacity upper bound $U_r$ by $U_r\leftarrow U_{u^*}$ and reduce to the second case. In fact, we can similify the argument of this case in \citep{an2015centrality}. We only need to verify the second condition in Definition \ref{def:transfer}. By the seond case, $y'_r=1$ always holds, i.e. $r\in F$. Since $F$ is an integral distance-$2$ transfer of $y$ by replacing $U_r$ by a smaller value $U_{u^*}$, we can check that the second condition in Definition \ref{def:transfer} still holds for the original capacity $U_r$.
\end{enumerate}
 \end{proof}

Thus, we directly have the following theorem by \cite{an2015centrality}.

\begin{theorem}
\label{lm:9app}
There is a polynomial time 9-approximation algorithm for the \knonuni\  problem. For the uniform capacity upper bound version, the ($L$,$U$,$k$)-\textsc{Center} problem admits a 6-approximation.
\end{theorem}

\begin{theorem}
\label{lm:11supapp}
There is a polynomial time 13-approximation algorithm for the \ksupnonuni\  problem. For the uniform capacity upper bound version, the ($L$,$U$,$k$)-\textsc{Supplier} problem admits a 9-approximation.
\end{theorem}

\topic{The case with outliers.}
For the (0,$\{U_u\}$,$k$,$p$)-\textsc{Supplier} problem with outliers, \cite{kociumaka2014constant} gave an algorithm to find an integral distance-24 transfer. W.l.o.g., they assume the optimal capacitated $k$-supplier value is 1. For each facility $u\in \calF$, define $deg(u)=|\{v\in \calC, \dist(u,v)\leq 1\}|$. They need to assume that $U_u\leq deg(u)$ by setting $U_u\leftarrow \min\{U_u, deg(u)\}$. The same as in Section \ref{sec:greedy}, we can still make this assumption. 
Another difference is that they reduce the problem to a tree instance via a data structure called \emph{skeleton}, see \citep[Definition 3]{kociumaka2014constant}. We need to verify the existence of a skeleton for the case that $L>0$. Fortunately, we can first modify $L=0$ and use the algorithm in \citep{kociumaka2014constant} to find a skeleton $S$. It is not hard to verify that $S$ is also a skeleton for $L$. Similarly, by \citep{kociumaka2014constant}, we have the following theorem.

\begin{theorem}
\label{lm:25app}
There is a polynomial time 25-approximation algorithm for the \ksupnonuniout\  problem and the ($L$,$\{U_u\}$,soft-$k$,$p$)-\textsc{Supplier} problem. For the uniform capacity upper bound version, the ($L$,$U$,$k$,$p$)-\textsc{Supplier} problem admits a 23-approximation, and the ($L$,$U$,soft-$k$,$p$)-\textsc{Supplier} problem admits a 13-approximation.
\end{theorem}

Note that taking $\calC = \calF = V$ shows that the $k$-supplier problem generalizes the $k$-center problem, and consequently gives the same approximation bounds for the latter.

\begin{corollary}
\label{lm:25app}
There is a polynomial time 25-approximation algorithm for the \knonuniout\  problem and the ($L$,$\{U_u\}$,soft-$k$,$p$)-\textsc{Center} problem. For the uniform capacity upper bound version, the ($L$,$U$,$k$,$p$)-\textsc{Center} problem admits a 23-approximation, and the ($L$,$U$,soft-$k$,$p$)-\textsc{Center} problem admits a 13-approximation.
\end{corollary}

\clearpage
\newpage
\bibliographystyle{authordate1}
\bibliography{lipics-v2016-sample-article}
\appendix

\section{Proof of Lemma \ref{lm:bipartite}}
\label{app:lemma4}

\begin{proof}
Suppose we have such a polynomial time algorithm $A$ as described in the lemma, and the optimal capacitated $k$-supplier value of the input \ksupbnonuniout\  instance is $r>0$. We want to give a polynomial time algorithm that outputs a feasible solution with a capacitated $k$-supplier value at most $\rho r$.

We first guess the value of $r$. Note that there are at most $|\calC|\times \calF|$ possible values for the optimal capacitated $k$-supplier value $r$, we can enumerate all of them. Assume that we know the value $r$. We connect each pair $(u,v)$ ($u\in\calF$, $v\in\calC$) if and only if $\dist(u,v)\leq r$, and obtain an undirected graph $G=(\calF\cup\calC;E)$. Suppose $G$ can be decomposed into different connected components  $G_i=(\calC_i\cup \calF_i, E_i)$ ($i=1,2,\cdots,m$). Since the optimal capacitated $k$-supplier value of the input \ksupbnonuniout\ instance is $r$, there must exist client sets $C_i\subseteq \calC$ of size $p_i\in\bbbz_{\geq 0}$ satisfying that $\sum_{i=1}^m p_i\geq p$, open facility sets $F_i\subset \calF_i$ of size $k_i$ satisfying that $\sum_{i=1}^m k_i=k$, and an assignment of every client $v\in C_i$ ($i=1,2,\cdots,m$) to an open facility $u\in F_i$ such that $\dist(u,v)\leq r$ and the capacity constraints are satisfied.

For each connected component $G_i$, consider the distance function $\dist_{G_i}: (\calC_i\cup \calF_i) \times (\calC_i\cup \calF_i) \rightarrow \R_{\geq 0}$ induced by $G_i$. Combining with the above argument, we have that the optimal capacitated $k_i$-supplier value of the induced ($\{L_u\}$,$\{U_u\}$,$k_i$,$p_i$)-\textsc{Supplier} instance induced by $G_i$ is at most 1. If we have the exact value of $k_i$ and $p_i$, we can find a feasible solution with a capacitated $k_i$-supplier value at most $\rho$ for this induced instance by the assumption in the lemma. Observe that for any $s_1,s_2\in \calC_i\cup \calF_i$, we have that the original distance $\dist(s_1,s_2)$ between $s_1$ and $s_2$ satisfies that $\dist(s_1,s_2)\leq r\cdot \dist_{G_i}(s_1,s_2)$. Thus, the union of all such feasible solutions of each $G_i$ is a $\rho$-approximation feasible solution of the original \ksupbnonuniout\ problem. We remain to show that we can find values $k_i,p_i$ in polynomial time.

Let the boolean value $A_i[k'][p']$ indicate whether the algorithm $A$ finds a feasible solution with capacitated $k'$-supplier value at most $\rho$ for the induced ($\{L_u\}$,$\{U_u\}$,$k_i$,$p_i$)-\textsc{Supplier} instance induced by $G_i$. Let the boolean value $K[i][k'][p']$ indicate whether there exists a feasible solution that opens $k'$ facilities in $\calF_1\cup \calF_2\cup \cdots\cup \calF_i$ to cover at least $p'$ clients in $\calC_1\cup \calC_2\cup \cdots\cup \calC_i$, satisfying all capacity constraints and the capacitated $k'$-supplier value (the original \ksupbnonuniout\ instance) is at most $\rho r$. Initially, let $K[1][k'][p']=A_1[k'][p']$. For $i>1$, we update
$$ K[i][k'][p']=\bigvee_{0\leq k^*\leq k'\atop 0\leq p^*\leq p'}(K[i-1][k'-k^*][p'-p^*]\wedge A_i[k^*][p^*]).$$
Note that $K[m][k][p]=\true$ since all $A_i[k_i][p_i]=\true$. Thus in polynomial time we can find values $k_i$ and $p_i$ ($i=1,2,\cdots,m$) such that $A_i[k_i][p_i]=\true$ and $\sum_{i=1}^m k_i=k,\sum_{i=1}^m p_i\geq p$.
\end{proof}

\end{document}